\documentclass[11pt, letter]{Liebert_Author}

\usepackage[none]{hyphenat}
\usepackage{graphicx}
\usepackage{amsmath,amsfonts,amssymb}
\newtheorem{definition}{Definition}
\newtheorem{example}{Example}
\newtheorem{proposition}{Proposition}
\renewcommand{\leq}{\leqslant}
\renewcommand{\geq}{\geqslant}
\usepackage{float}
\usepackage{xcolor}
\newcommand{\calP}{\mathcal{P}}
\newcommand{\supp}{\mbox{Supp}}
\newcommand{\kRF}{\tiny  k\mbox{\rm -RF}}
\newcommand{\OneRF}{\tiny  1\mbox{\rm -RF}}

\title{The $k$-Robinson-Foulds Dissimilarity Measures for Comparison of Labeled  Trees} 

\author{Elahe Khayatian$^{1}$, Gabriel Valiente$^{2}$, Louxin Zhang$^{1\ast}$\\
{$^{1}$Department of Mathematics, National University of Singapore,}\\
{Singapore 119076}\\
{$^{2}$Department of Computer Science,}\\
{Technical University of Catalonia, E-08034 Barcelona, Spain}\\
{$^\ast$Corresponding author:
{E-mail: matzlx@nus.edu.sg.}}
}
\date{}

\begin{document} 
\maketitle 

\begin{abstract}

{
Understanding the mutational history of tumor cells is a critical endeavor in unraveling the mechanisms underlying cancer. Since the modeling of tumor cell evolution employs labeled trees, researchers are motivated to develop different methods to assess and compare mutation trees and other labeled trees. While the Robinson-Foulds distance is a widely utilized metric for comparing phylogenetic trees, its applicability to labeled trees reveals certain limitations. This paper introduces the $k$-Robinson-Foulds dissimilarity measures, tailored to address the challenges of labeled tree comparison. The Robinson-Foulds distance is succinctly expressed as $n$-RF in the space of
labeled trees with $n$ nodes. Like the Robinson-Foulds distance, the $k$-Robinson-Foulds is a pseudometric for multiset-labeled trees and becomes a metric in the space of 1-labeled trees. By setting $k$ to a small value, the
$k$-Robinson-Foulds dissimilarity can capture analogous local regions in two labeled trees with different size or different labels.}
\end{abstract}

\section{Introduction}
\label{sec 1}

Trees in biology are a fundamental concept as they depict the evolutionary history of entities. These entities may consist of organisms, species, proteins, genes or genomes. Trees are also useful for healthcare analysis and medical diagnosis. Introducing different kinds of tree models has given rise to the question about how these models can be efficiently compared for evaluation. This question has led to defining a robust dissimilarity measure in the space of targeted trees. For example, mutation/clonal trees are introduced to model tumor evolution, in which nodes represent cellular populations and are labeled by the gene mutations carried by those populations~\citep{8Karpov,schwartz2017evolution}.  The progression of tumors varies among different patients; additionally, information about such variations is significant for cancer treatment. Therefore, dissimilarity measures for mutation trees have become a focus of recent research (\cite{4Dinardo,1John_Zhang2019,llabres.ea:2021,8Karpov}).

In prior studies on species trees, several measures have been introduced to compare two phylogenetic trees. Some examples of such distances are Robinson-Foulds distance (RF)~\citep{2RobinsonF}, Nearest-Neighbor Interchange (NNI)~\citep{li1996,10Robinson}, Quartet distance~\citep{11Estabrook}, and Path distance~\citep{3Steel,12williams1971}. Although these distances have been widely used for phylogenetic trees, they are defined based on the assumption that the involved trees have the same label sets. Moreover, only leaves of phylogenetic trees are labeled. Thus, these distances are not useful for comparing trees with different label sets or trees in which all the nodes are labeled.

\subsection{Related Work on Comparison of Labeled Trees}

To get around some limitations of the above-mentioned distances in the comparison of mutation trees, researchers have introduced new measures for mutation trees. Some of these measures are Common Ancestor Set distance (CASet)~\citep{4Dinardo}, Distinctly Inherited Set Comparison distance (DISC)~\citep{4Dinardo}, and Multi-Labeled Tree Dissimilarity measure (MLTD)~\citep{8Karpov}. 
{
While these distance measures enable efficient comparison of clonal trees, they are defined based on the assumption that  mutations cannot occur more than once and  mutations will not be lost in the course of tumor evolution. As a result, these metrics exhibit multiple limitations when applied to the comparison of trees used to model complex tumor evolution, wherein mutations may indeed occur multiple times and subsequently be lost.
}

In addition to the three measures mentioned above, a group of other dissimilarity measures have been introduced for the comparison of mutation trees, including Parent-Child Distance~\citep{9Govek} and Ancestor-Descendant Distance~\citep{9Govek}. These measures are metric only in the space of `1-mutation' trees, in which each node is labeled by exactly one mutation. These distances are again defined based on the assumption that mentioned above.

Additionally,  there are other measures for mutation trees, defined based on the generalization methods. In such methods, researchers aim to extend the definition of an existing distance, which mostly used to compare phylogenetic trees, to mutation trees. For example, the generalized  Nearest Neighbour Interchange (gNNI) (\cite{1John_Zhang2019}) is defined by some minor modifications of NNI, which was first defined for the comparison of phylogenetic trees. The  other example is the Path Distance (\cite{9Govek}) which was first defined for phylogenetic trees comparison. Although these measure are applicable to mutation trees, they are only well defined for mutation trees with the same label sets (\cite{1John_Zhang2019, 9Govek}).

Apart from the measures mentioned above, another recently proposed distance is the generalized RF distance (GRF)~\citep{llabres.ea:2020,llabres.ea:2021}. This measure allows for the comparison of phylogenetic trees, phylogenetic networks, mutation and clonal trees. An important point about this distance is that its value depends on the intersection between clusters or clones of trees. However, this intersection does not contribute to the RF distance. In fact, if two clusters or clones of two trees are different, their contribution to the RF distance is 1; otherwise, it is 0. Hence, the generalized RF distance has a better resolution than the RF distance. However, it is defined based on the assumption that two distinct nodes in each tree are labeled by two disjoint sets~\citep{llabres.ea:2020}.

There are some other generalizations of the RF distance, such as Bourque distance~\citep{1John_Zhang2019}. This distance is effective for comparing mutation trees whose nodes are labeled by non-empty sets and has linear time complexity. However, like the above distances, it does not allow for multiple occurrences of mutations during the tumor history~\citep{1John_Zhang2019}. Other generalization of the RF distance have also been proposed for gene trees~\citep{Lafond,briand2022linear}.

The above-mentioned measures are not able to quantify similarity or difference of some tree models. Two instances of such models are the Dollo~\citep{16Farris} and the Camin-Sokal model~\citep{17camin1965method}. The reason behind the inadequacy of the measures for these models is that it is possible for mutations to get lost after they are gained in the Dollo model, and a mutation can occur more than once during the tumor history in the Camin-Sokal model~\citep{llabres.ea:2020}. Hence, some measures are needed to mitigate the problem. To the best of our knowledge, Triplet-based Distance~\citep{13Ciccolella} is the only measure introduced so far to resolve the issue. The distance is useful for comparing mutation trees whose nodes are labeled by non-empty sets; additionally, it allows for multiple occurrences of mutations during the tumor history and losing a mutation after it is gained~\citep{13Ciccolella}. Thus, the measure is applicable to the broader family of trees in which two nodes of a tree may have non-disjoint sets of labels. Nevertheless, it is not able to compare those trees in which there is a node whose label has more than one copy of a mutation.

Although no tree model has been introduced so far that allows for more than one copy of a mutation in the label of a single node, current models can be extended to deal with copy number of mutations. For example, the constrained $k$-Dollo model~\citep{19Palash} takes the variant read count and the total read count of each mutation in each cell, derived from single-cell DNA sequencing data, as input; then, based on three thresholds for the variant read count, the total read count, and the variant allele frequency, it decides whether a mutation is present or absent in a cell or it is missing~\citep{19Palash}. Alternatively, the model can consider the exact frequency numbers to show the multiplicity of each mutation in each cell. This motivates us to introduce new distances 
that can be used to compare pairs of labeled trees whose nodes are labeled by non-empty multisets.

\subsection{Our Contributions to Tree Comparison}


In this paper, we present $k$-RF dissimilarity measures designed for the comparison of labeled trees. They are first defined for 1-labeled trees (Section~\ref{sec 3}). Subsequently, we extend these measures to multiset-labeled trees (Section~\ref{sec 6}). We delve into the mathematical properties of the $k$-RF measures in  Sections~\ref{sec_4} and \ref{sec 6}. In particular, $k$-RF is a metric for 1-labeled trees.  We also assess the validity of the $k$-RF measures through comparisons with CASet, DISC, and GRF (Section~\ref{sec 6}), and the evaluation of their performance in the context of tree clustering (Section~\ref{sec_5}).

\section{Concepts and Notations}
\label{sec 2}

A graph consists of a set of nodes and a set of edges that are each an unordered pair of distinct nodes, whereas a directed graph consists of a set of nodes and a set of directed edges that are each an ordered pair of distinct nodes.

Let $G$ be a (directed) graph. We use $V(G)$ and $E(G)$ to denote its node and edge set, respectively. 
If $G$ is undirected,  $(u, v)$ will still be used to denote an edge between $u$ and $v$ with the understanding that $(u, v)=(v, u)$.  Let $u, v \in V(G)$. 
 If  $(u, v) \in E(G)$, we say that $u$ and $v$ are adjacent,  the edge $(u, v)$ is incident to $u$ and $v$, or $u$ and $v$ are 
two endpoints of $(u, v)$.

The degree of $v$ is defined as the number of edges incident to $v$. 
In addition, if $G$ is directed, the \emph{indegree} and \emph{outdegree} of $v$ are defined as the number of edges $(x, y)$ such that $y = v$ and $x = v$, respectively. 
 The nodes of degree 1 are called the \emph{leaves} in an undirected graph, whereas the nodes of  indegree 1 and outdegree 0 are called the \emph{leaves} in a directed graph. 
We use $\mathit{Leaf}(G)$ to denote the leaf set for $G$. Non-leaf nodes are called \emph{internal nodes}. 

A \emph{path} of length $k$ from $u$ to $v$ consists of a sequence of nodes $u_0, u_1, \ldots, u_k$ such that $u_0 = u$, $u_k = v$ and $(u_{i-1}, u_{i}) \in E(G)$ for  $i=1, 2, \cdots, k$. 
 The \emph{distance} from $u$ to $v$, denoted as $d_G(u, v)$, is the length of the shortest paths from $u$ to $v$, 
and it is set to $\infty$ if there is no path from $u$ to $v$.
Note that if $G$ is undirected, $d_G(u, v) = d_G(v, u)$ for all $u, v \in V(G)$. The \emph{diameter} of $G$, denoted as $\textrm{diam}(G)$, is defined as $\max_{u, v \in V(G)} d_G(u, v)$. 
If $G$ is directed, 
its diameter is defined as  the diameter of its undirected version that has the node set $V(G)$ and edge set $ E(G)\cup \{(u, v) \mid (v, u)\in E(G) \}$.

\subsection{Trees}

A tree $T$ is a graph in which there is exactly one path from every node to any other node. It is \emph{binary} if every internal node is of degree 3. It is a \emph{line tree} if every internal node is of degree 2. Each line tree has exactly two leaves.

\subsection{Rooted Trees}

A rooted tree is a directed tree with a specific root node where the edges are oriented away from the root. In a rooted tree, there is exactly one edge entering $u$ for every non-root node $u$, and thus there is a unique path from its root to every other node.

Let $T$ be a rooted tree and $u, v \in V(T)$. If $(u, v) \in E(T)$, $v$ is called a child of $u$ and $u$ is called the parent of $v$. In general, for $u \neq v$, if $u$ belongs to the unique path from $\mathit{root}(T)$ to $v$, $v$ is said to be a descendant of $u$, and $u$ is said to be an ancestor of $v$. We use $C_T(u)$, $A_T(u)$ and $D_T(u)$ to denote the set of all children, ancestors and descendants of $u$, respectively. Note that $u \notin A_T(u)$ and $u \notin D_T(u)$.


A binary rooted tree is a rooted tree in which the root is of indegree 0 and outdegree 2, and every other internal node is of indegree 1 and outdegree~2.
A rooted line tree is a rooted tree in which each internal node has only one child.
A rooted caterpillar tree is a rooted tree in which every internal node has at most one child that is internal.

\subsection{Labeled Trees}
\label{sec23}

Let $L$ be a set and $\mathbb{P}(L)$ be the set of all subsets of $L$. A tree or rooted tree $T$ is labeled with the subsets of $L$ if $T$ is equipped with a function $\ell : V(T) \to \mathbb{P}(L)$ such that $\cup_{v \in V(T)} \ell(v) = L$, and $\ell(v) \neq \emptyset$ for every $v \in V(T)$. In particular, if $\ell(v)$ contains exactly one element for each $v \in V(T)$ and $\ell$ is one-to-one, $T$ is said to be a $1$-labeled tree on $L$. 
In addition, if $T$ is 1-labeled on $L$, then for $C \subseteq V(T)$, $L(C)$ is defined as $L(C) = \{ x \in L \mid \exists w \in C : \ell(w) = \{x\} \}$.


\subsection{Phylogenetic and Mutation Trees}

Let $X$ be a finite taxon set. A phylogenetic tree (respectively, rooted phylogenetic tree) on $X$ is a binary tree (respectively, binary rooted tree) in which the leaves are uniquely labeled by the taxa of $X$ and the internal nodes are unlabeled.

A mutation tree on a set M of mutated genes is a rooted tree in which nodes are labeled with nonempty subsets of $M$.

\subsection{Dissimilarity Measures for Trees}

Let $\mathcal{T}$ be a set of trees. A dissimilarity measure on $\mathcal{T}$ is a symmetric real function $d : \mathcal{T} \times \mathcal{T} \to \mathbb{R}^{\geq 0}$. A dissimilarity measure should capture  the idea that the more similar two trees are, the lower their measure value is. A pseudometric on $\mathcal{T}$ is a dissimilarity measure that satisfies the triangle inequality condition. Finally, a metric (distance) on $\mathcal{T}$ is a pseudometric $d$ such that $d(S, T) \neq  0$ unless $S$ and $T$ are the same tree.

\section{The \texorpdfstring{$k$-RF}{k-RF} Measure for 1-Labeled Trees}
\label{sec 3}

In this section, we first recall the definition of the RF distance  and then present $k$-RF dissimilarity measures for 1-labeled trees for arbitrary $k$.

\subsection{The \texorpdfstring{$k$-RF}{k-RF} Measure for 1-Labeled Unrooted Trees}
\label{subsec 3.1}

Let $X$ be a set of labels and let $T$ be a 1-labeled tree over $X$. Each $e = (u, v) \in E(T)$ induces a pair of label subsets on $X$:
\begin{eqnarray}
P_{T}(e) = \left\{ L(B_e(u)), L(B_e(v)) \right\}, 
\label{eqn1_sec31} \\
B_e(u) = \{ w \mid d_T(w, u) < d_T(w, v) \}, 
\nonumber \\
B_e(v) = \{ w \mid d_T(w, v) < d_T(w, u) \}.
\label{eqn2_sec31}
\end{eqnarray}
We further define:
\begin{equation}
\calP(T) = \{ P_{T}(e) \mid e \in E(T) \}.
\label{eqn3_sec31}
\end{equation}
%
The RF distance of two 1-labeled trees $S$ and $T$ is defined as:
\begin{equation}
d_{RF}(S, T) = \vert \calP(S) \bigtriangleup \calP(T) \vert.
\label{eqn4_sec31}
\end{equation}

\begin{figure}[!b]
\centering
\includegraphics[scale=0.6]{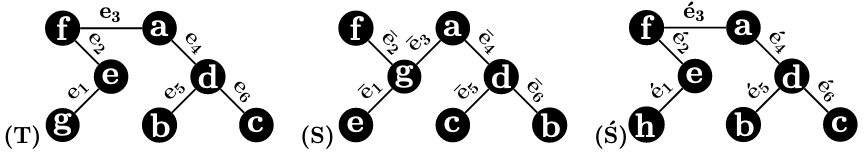}
\caption{\label{fig:1-1} 
Three 1-labeled trees in Example~\ref{example.1} to illustrate that the Robinson-Foulds 
distance exhibits a heavy penalty against trees with different labels. Although $T$ and $\acute{S}$ is only different
 in labelling one node, the RF distance is 4 for $S$ and $T$, but 12 for  $\acute{S}$ and $T$.
}
\end{figure}

\begin{example}
\label{example.1}
Consider the three 1-labeled trees in Figure~\ref{fig:1-1}. We have $d_{\mbox{\tiny RF}}(S, T) = 4$, because $P_T(e_1)$ to $ P_T(e_6)$ are:
\begin{eqnarray*}
   \left\{\{a, b, c, d, e, f\},\; \{g\}\right\},
& \left\{\{a, b, c, d, f\},\; \{e, g\}\right\},
& \left\{\{a, b, c, d\},\;\{e, f, g\}\right\},\\
   \left\{\{a, e, f, g\},\; \{b, c, d\}\right\}, 
& \left\{\{a, c, d, e, f, g\},\; \{b\}\right\},
& \left\{\{a, b, d, e, f, g\},\; \{c\}\right\},
\end{eqnarray*}
respectively, whereas $P_S(\bar{e}_1)$ to $ P_S(\bar{e}_6)$ are:
\begin{eqnarray*}
   \left\{\{a, b, c, d, f, g\},\; \{e\}\right\},
& \left\{\{a, b, c, d, e, g\},\; \{f\}\right\},
& \left\{\{a, b, c, d\},\; \{e, f, g\}\right\},\\
   \left\{\{a, e, f, g\},\; \{b, c, d\}\right\},
& \{\{a, b, d, e, f, g\},\; \{c\}\},  
&\{\{a, c, d, e, f, g\},\; \{b\}\},
\end{eqnarray*}
respectively. However, $d_{\mbox{\tiny RF}}(\acute{S}, T) = 12$, even if $T$ and $\acute{S}$ have the same topology and only one node is labeled differently.
\end{example}

The above example indicates that the RF cannot capture local similarity (and difference) for 1-labeled trees if they have different labels. One popular dissimilarity measure for sets is the Jaccard distance. It is obtained by dividing the size of the symmetric difference of two sets by the size of their union. Two 1-labeled trees are identical if and only if they have the same set of  edges. Therefore, we propose to use $\vert E(S) \bigtriangleup E(T) \vert$ and its generalization to measure the dissimilarity for 1-labeled trees $S$ and $T$.

Let $k$ be a non-negative integer and let $T$ be a $1$-labeled tree. Each edge $e = (u, v)$ induces the following pair of subsets of labels:
\begin{eqnarray}
&& P_{T}(e, k) = \{ L(B_e(u, k)), L(B_e(v, k)) \},
\label{eqn5_sec32} \\
&& B_e(x, k) = \{ w \in B_e(x) \mid d_T(w, x) \leq k \}, \;\; x=u, v. \nonumber
\end{eqnarray}
Clearly, $B_e(u, \infty) = B_e(u)$ and $B_e(u, 0) = \{ u \}$. We further define:
\begin{equation}
\calP_{k}(T) = \{ P_{T}(e, k) \mid e \in E(T) \}.
\label{eqn6_sec32}
\end{equation}

\begin{definition}
\label{k-RF-1}
Let $k \geq 0$ and let $S$ and $T$ be two 1-labeled trees. The $k$-RF dissimilarity score  of $S$ and $T$ is defined as:
\begin{equation}
d_{\kRF}(S, T) = |\calP_{k}(S) \bigtriangleup\calP_{k}(T)|.
\label{eqn7_sec32}
\end{equation}
\end{definition}

\begin{example}
Continuing with Example~\ref{example.1}, we have $d_{\OneRF}(\acute{S}, T) = 4$, as $P_{T}(e_i, 1)$ for $1 \leq i \leq 6$ are:
\begin{center}
\begin{tabular}{lllll}
   $\{\{g\}, \{ e, f\}\}$,
&& $\{\{e, g\}, \{a, f\}\}$,
&& $\{\{e, f\}, \{a, d\}\}$,\\
   $\{\{a, f\}, \{b, c, d\}\}$,
&& $\{\{b\}, \{a, c, d\}\}$,
&& $\{\{c\}, \{a, b, d\}\}$,
\end{tabular}
\end{center}
respectively, and $P_{\acute{S}}(\acute{e}_i, 1)$ for $1 \leq i \leq 6$ are:
\begin{center}
\begin{tabular}{lllll}
   $\{\{h\}, \{e, f\}\}$,
&& $\{\{e, h\}, \{a, f\}\}$,
&& $\{\{e, f\}, \{a, d\}\}$, \\
   $\{\{a, f\}, \{b, c, d\}\}$,
&& $\{\{b\}, \{a, c, d\}\}$,
&& $\{\{c\}, \{a, b, d\}\}$, 
\end{tabular}
\end{center}
respectively. We also have $d_{\OneRF}(S, T) = 8$. Thus,  $1$-RF  captures the difference of the trees better than the RF distance.
\end{example}

\subsection{The \texorpdfstring{$k$-RF}{k-RF} Measure for 1-Labeled Rooted Trees}
\label{subsec 3.4}

Let $k \geq 0$ be an integer and let $T$ be a $1$-labeled rooted tree. For a node $w \in V(T)$, we define $B_k(w)$ and $D_k(w)$ as:
\begin{eqnarray}
&& B_k(w) = \{ x \in V(T) \mid \exists y \in A_T(w) \cup \{ w \} : d(y, w) + d(y, x) \leq k \},
\label{def_Bk}\\
&& D_k(w) = \{ w \} \cup \{ x \in D_T(w) \mid d(w, x) \leq k \}.
\label{def_Dk} 
\end{eqnarray}

For each $e = (u, v) \in E(T)$, we define $P_{T}(e, k)$ as the following ordered pair od two label subsets:
\begin{equation}
P_{T}(e, k) = (L(D_{k}(v)), L(B_k(u) \setminus D_k(v))).
\label{eqn11_sec34}
\end{equation}
Here, the first subset of $P_T(e,k)$ contains the labels of the descendants that are at distance at most $k$ from $v$, whereas the second subset contains the labels of the other nodes around the edge $e$ within a distance of $k$. Next, we define:
{
\begin{equation}
\mathcal{P}_{k}(T) = \{ P_{T}(e, k) \mid e\in E(T) \}.
\label{eqn12_sec34}
\end{equation}
} 

\begin{definition}
\label{k-root-1}
Let $k\geqslant 0$ and let $S$ and $T$ be two 1-labeled rooted trees. Then, the $k$-RF dissimilarity between $S$ and $T$ is defined as:
\begin{equation}
d_{\kRF}(S, T) = |\mathcal{P}_{k}(S)  \bigtriangleup \mathcal{P}_{k}(T)|.
\label{eqn13_sec34}
\end{equation}
\end{definition}

\begin{example}
\label{example.2}
Consider the two 1-labeled rooted trees $S$ and $T$ in Figure~\ref{fig:amaa-1}. $P_T(e_i, 1)$ ($1 \leq i \leq 7$) are the following ordered pairs of label subsets:
$$
\begin{array}{lllllll}
   (\{f, h\}, \{b, d\}),
&& (\{c, f, g\}, \{ b, h\}),
&& (\{c\}, \{f, g, h\}),
&& (\{g\}, \{c, f, h\}), \\
   (\{a, d, e\}, \{b, h\}),
&& (\{a\}, \{b, d, e\}),
&& (\{e\}, \{a, b, d\}).
&& 
\end{array}
$$
 $P_S(\bar{e}_i, 1)$ ($1 \leq i \leq 7$) are the following ordered pairs of label subsets:
$$
\begin{array}{lllllll}
   (\{b, d\}, \{c, f\}),
&& (\{a, d, e\}, \{b, c\}),
&& (\{a\}, \{b, d, e\}),
&& (\{e\}, \{a, b, d\}),\\
   (\{f, g, h\}, \{b, c\}),
&& (\{g\}, \{c, f, h\}),
&& (\{h\}, \{c, f, g\}).
&&
\end{array}
$$
Therefore, $d_{\OneRF}(S, T)=8$.
\end{example}

\begin{figure}[!t]
\centering
\includegraphics[scale=0.5]{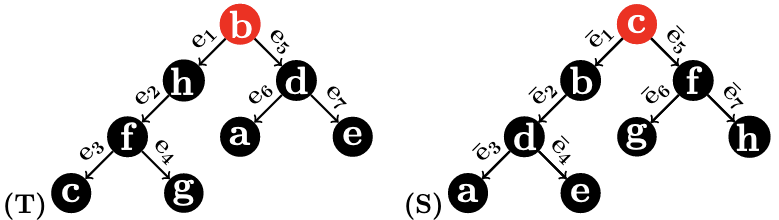}
\caption{\label{fig:amaa-1} Two 1-labeled rooted trees used to illustrate the  $1$-RF in Example~\ref{example.2}.}
\end{figure}

\section{Characterization of \texorpdfstring{$k$-RF}{k-RF} for 1-Labeled Trees}
\label{sec_4}

In order to evaluate  $k$-RF, we first provide the mathematical properties of the $k$-RF.  We then present  experimental results on the frequency distribution of these measures.

\subsection{Mathematical Properties}

\begin{proposition}
\label{main-1}
Let $S$ and $T$ be two 1-labeled trees. 

 {\rm (a)}  Let $|L(S) \cap L(T)| \leq 2$ and $|E(T)| \geq 2$. For any $k \geq 1$, $d_{\kRF}(S, T) = |E(S)| + |E(T)|$.
 
{\rm (b)}  Assume that $L(S) \neq L(T)$. For $k < \min \{ \textrm{diam}(T), \textrm{diam}(S) \}$, \\
$k + 1 \leq d_{\kRF}(S, T) \leq |E(S)| + |E(T)|$. In addition, the second inequality become equality if $k \geq \min \{ \textrm{diam}(T), \textrm{diam}(S) \}$ and $|L(S)| = |L(T)|$

{\rm (c)} Renaming each node with its label, we have $d_{\tiny 0\mbox{\rm -RF}}(S, T) = \vert E(S) \bigtriangleup E(T) \vert$.

{\rm (d)}  If $k \geq \max \{ \textrm{diam}(S), \textrm{diam}(T) \} - 1$, then $d_{\kRF}(S, T) = d_{RF}(S, T)$.
\end{proposition}

\begin{proof}
(a) Note that if $k \geq 1$ and $|E(T)| \geq 2$,  each $P_{T}(e, k)$ involves at least three labels.  
If   $L(S)$ and $ L(T)$ have  only two common elements,   $P_{T}(e, k) \neq P_{S}(\acute{e}, k)$   for every $e \in E(T)$ and $ \acute{e} \in E(S)$. Thus, we have $\calP_{k}(S) \cap \calP_{k}(T) = \emptyset$, implying that 
$
d_{\kRF}(S, T) = \vert \calP_{k}(S)  \bigtriangleup \calP_{k}(T) \vert  =\vert \calP_{k} (T) \vert + \vert \calP_{k} (S) =
|E(S)| + |E(T)|.
$

(b)~  The second inequality follows from that 
$d_{\kRF}(S, T) = \vert \calP_{k}(S)  \bigtriangleup \calP_{k}(T) \vert 
\leq \vert \calP_{k} (T) \vert + \vert \calP_{k} (S) \vert $ and $\vert \calP_{k}(X)\vert=\vert E(X)\vert $ for $X=S, T$. We prove the first inequality as follows.

Let $k < \min \{ \mbox{diam}(T), \mbox{diam}(S) \}$. 
Since $S$ and $T$ are 1-labeled, we identify a node with its label in the trees. 
 Without loss of generality, 
 we may assume $v\in V(T)\setminus V(S)$. 
 Define 
 ${\cal N}^{(k)}_T(v) = \{ u \;\vert \; d_T(u, v)\leq k\}$.
 
 If ${\cal N}^{(k)}_T(v) =V(T)$, then, 
 $\vert {\cal N}^{(k)}_T(v) \vert 
 =\vert V(T) \vert \geq \mbox{diam}(T) +1\geq k+2$, as $k < \mbox{diam}(T)$. 
 This also implies that for every $(x, y)\in E(T)$, $d_T(v, x)\leq k$ and 
 $d_T(v, y)\leq k$. 
 
If ${\cal N}^{(k)}_T(v) \neq V(T)$, there exists at least a node $w$ that is  $k+1$ or more distance away from $v$. Since $T$ is connected,  we let $P(v, w)$ be a path from $v$ and $w$ with the smallest length. 
Clearly, the first $k+1$ nodes in $P(v, w)$ (including $v$) are all in ${\cal N}^{(k)}_T(v)$, i.e. at least one end of the first $k+1$ edges of $P(v, w)$ are found in 
${\cal N}^{(k)}_T(v)$.

 In summary, we have proved that there are at least $k+1$ edges $(x, y)$ such that either $d_T(v, x)\leq k$ or $d_T(v, y)\leq k$. For each of these edges $e$,  $v$ appears in at least one label subset of $P_T(e, k)$  and thus
 $P_T(e, k)\not\in \mathcal{P}_k(S)$. Therefore,
 $d_{\kRF}(S, T) \geq \vert  \mathcal{P}_k(T)\setminus  \mathcal{P}_k(S)\vert \geq k+1$.



If $|L(S)| = |L(T)|$ and $k \geq \min \{ \textrm{diam}(T), \textrm{diam}(S) \}$,
then, ${\cal N}^{(k)}_T(v) =V(T)$. 
Therefore,  
the induced pair $P_T(e, k)$ contains $v$ for every edge $e$ of $T$. On the other hand, 
the induced pair $P_S(e, k)$ does not contain $v$ for each edge $e$ of $S$. Thus, $\mathcal{P}_k(S)\cap \mathcal{P}_k(T)=\emptyset$ and $d_{\kRF}(S, T)=
\vert \mathcal{P}_k(S)\vert +\vert \mathcal{P}_k(T)\vert = \vert E(S)\vert + \vert E(T) \vert$.


(c)~  Note that we may represent each node of a 1-labeled tree with its unique label. As a result, $P_{T}(e, 0) = e$ and $P_{S}(\bar{e}, 0) = e$ for $e \in E(T)$ and $\bar{e} \in E(S)$. Thus, \textbf{(c)} follows.

(d)~ It follows from the definition of the $k$-RF.
\end{proof}

\begin{lemma}
\label{important-1}
Let $k \geq 0$ be an integer.  $k$-RF  satisfies the non-negativity, symmetry and triangle inequality conditions.
\end{lemma}

\begin{proof}
Let $k\geq 0$.
The non-negativity and symmetry conditions are trivial.
%
The triangle inequality
$ d_{\kRF}(T_{1},T_{2}) \leq d_{\kRF}(T_{1},T_{3}) + d_{\kRF}(T_{3},T_{2})
$ 
is derived from the inequality
%
$ \calP_{k}(T_1) \bigtriangleup \calP_{k}(T_2) \subseteq (\calP_{k}(T_1) \bigtriangleup \calP_{k}(T_3)) \cup (\calP_{k}(T_3) \bigtriangleup \calP_{k}(T_2))
$
for any three 1-labeled trees $T_{1}, T_{2}, T_3$.
\end{proof}

\begin{remark}
\label{rem-1}
Proposition~\ref{main-1} and Lemma~\ref{important-1} can be proved in the same manner for 1-labeled rooted trees.
\end{remark}

%
%
%
%

\begin{proposition}
\label{0-metric-root}
The  0-RF is a metric on the space of all 1-labeled rooted trees. 
\end{proposition}

\begin{proof}
Let $S$ and $T$ be two 1-labeled rooted trees. By Remark~\ref{rem-1}, it is enough to show that  $S$ and $T$ are identical if $d_{\tiny 0\mbox{\rm -RF}}(S,T) = 0$.
By identifying a node with its label in $S$ and $T$, we obtain that
 $\mathcal P_{0}(S)=E(S)$ and $\mathcal P_{0}(T)=E(T)$. If   
 $d_{\tiny 0\mbox{\rm -RF}}(S,T) = 0$,  
$\vert E(T) \bigtriangleup E(S)\vert =0$ and thus $E(T)=E(S)$, i.e. $S$ and $T$ are identical. 
\end{proof}

\begin{lemma} 
Let $T$ be a 1-labeled rooted tree with at least two nodes and  $\mathcal L$ be a subset of $\mathit{Leaf}(T)$.  Define $T'$ to be the subtree obtained by the removal of all the leaves of ${\mathcal L}$.
 Then, for any $k$, 
\begin{eqnarray} {\mathcal P}_{k}(T') = 
\{ (X\setminus \mathcal{L}, Y\setminus \mathcal{L})\; \;\vert\; (X, Y)\in {\mathcal P}_{k}(T) \; \& \; X \cap \mathcal{L}\neq \emptyset\}.
\label{eqn13_new}
\end{eqnarray}
\end{lemma} 
\begin{proof}
Since $T$ is 1-labeled, we identify a node of $T$ with its label in the following discussion. With this convention, for any subset $S$ of nodes, $L(S)=S$. 
    
    Let $\bar{E}(T)$ denote the subset of edges incident to a leaf of $\mathcal{L}$, i.e., 
    $\bar{E}(T)=\{(x, y) \in E(T) \;\vert\; y\in \mathcal{L}\}$. Then, 
    $$V(T)= V(T')\uplus \mathcal{L},\;\;\ 
    E(T)=E(T')\uplus \bar{E}(T).
    $$

 If $(u, v)\in \bar{E}(T)$, 
   $v\in \mathcal{L}\subseteq \mathit{Leaf}(T)$ and thus $ D_k(v)=\{v\}\subseteq \mathcal{L}.$

   For an edge $e=(u, v)\in E(T')$, 
   $P_T(e, k)=(D_k(v), B_k(u)\setminus D_k(v))$.
   By Eqn.~(\ref{def_Bk}) and ~(\ref{def_Dk}, 
   \begin{eqnarray*}
      D_k(v) &=& D_k(v)\cap V(T') \uplus D_k(v)\cap \mathcal{L}, \\
      B_k(u)\setminus D_k(v)
        &=& 
        [(B_k(v)\setminus D_k(v)) \cap V(T')] \uplus [(B_k(v)\setminus D_k(v))\cap \mathcal{L}]\\
        &=& [(B_k(v)\cap V(T')]\setminus [D_k(v)\cap V(T')] 
          \uplus (B_k(v)\setminus D_k(v))\cap \mathcal{L}
   \end{eqnarray*}
  If $(u, v)\in E(T')$,  $  D_k(v)\setminus \mathcal{L}
    = D_k(v)\cap V(T') \neq \emptyset$
    and \\
 $(B_k(v)\setminus D_k(v))\setminus \mathcal{L}= (B_k(v)\cap V(T')]\setminus [D_k(v)\cap V(T')].$
    Therefore, \\
    $\left(D_k(v)\setminus \mathcal{L}, 
    (B_k(v)\setminus D_k(v))\setminus \mathcal{L}\right) = P_{T'}(e, k).$
    
  This has proved Eqn.~(\ref{eqn13_new}).
\end{proof}

\begin{proposition}
\label{important1}
Let $k \geq 1$ be an integer.  The $k$-RF  is a metric in the space of all 1-labeled rooted trees.
\end{proposition}

\begin{proof}
Let $S$ and $T$ be two 1-labeled rooted trees. By  Remark~\ref{rem-1}, it is enough to show that  $d_{\kRF}(S, T) = 0$ (equivalently, $\mathcal P_{k}(T) = \mathcal P_{k}(S)$) implies that  $S$ and $T$ are identical. 
To this end, we prove that $E(T)$ can be uniquely determined by $\mathcal P_{k}(T)$ using mathematical induction. 

Since $\vert E(T)\vert =\vert {\mathcal P}_{k}(T)\vert$, $T$ is a single node if and only if 
$E(T)$ is empty if and only ${\mathcal P}_{k}(T)$ is empty. Therefore, the single-node tree is uniquely determined  by ${\mathcal P}_{k}(T)$. 

 Assume  $S$ is uniquely determined by  ${\mathcal P}_{k}(S)$ for arbitrary 1-labeled tree $S$ such that $\vert V(S)\vert < k$. Consider a 1-labeled tree $T$ such that $\vert V(S)\vert = k$. 

For a leaf $v\in \mathit{Leaf}(T)$, there is a unique edge 
 $e=(u, v)$ entering $v$. Note that $k\geq 1$. Since 
 $D_k(v)=\{v\}$ if and only if $v$ is a leaf, we can identify $v$ from $P_T(e, k)=(P_1, P_2)\in {\mathcal P}_{k}(T)$  such that $P_1=\{v\}.$  

 For $v\in V(T)\setminus \mathit{Leaf}(T)$, there is  a unique edge $e=(u, v)$ entering $v$. Since $k\geq 1$, 
 the children of $v$ are all a leaf if and only if 
 $D_k(v)=\{v\} \cup C_T(u)$ if and only if 
 $D_K(v)\setminus \mathit{Leaf}(T)=\{v\}$. Therefore, we can identify $v$ whose children are all leaves from the ordered pairs $(P_1, P_2) \in {\mathcal P}_{k}(T)$ such that $ P_1\setminus \mathit{Leaf}(T)$  contains only $v$. 
 
 Let $V'$ be the set of all nodes whose children are just leaves and  $D_T(V')=\cup_{x\in V'}C_T(x)$.
 Since $V'$ is nonempty, $D_T(V')\neq \emptyset$ . 
 Define 
 $E'(T)=\{(x, y)\in E(T) \;\;\vert\;\; x\in V' , y\in D_T(V')\}$. 

 For the tree $T'$ obtained from $T$ by the removal of the leaves of $D_T(V')$,
$\vert V(T') \vert =\vert V(S)\vert -\vert D_T(V') \vert < k.$
By Eqn.~(\ref{eqn13_new}), ${\mathcal P}_{k}(T')$ can be efficiently constructed from  ${\mathcal P}_{k}(T)$.
By the induction hypothesis, $E(T')$ is uniquely determined by ${\mathcal P}_{k}(T')$. 
As a result, $E(T)=E(T')\cup E'(T)$ is determined.

This concludes the proof.
\end{proof}

\begin{corollary}
Let $k \geqslant 0$. The $k$-RF is  a metric in the space of all $1$-labeled trees.
\end{corollary}

\begin{proof}
If $k = 0$, the statement follows from the same proof as for Proposition~\ref{0-metric-root}. Now, let $S$ and $T$ be two 1-labeled trees and $k \geq 1$. By Lemma~\ref{important-1}, it is enough to show that if $d_{\kRF}(S,T) = 0$ (equivalently, $\mathcal P_{k}(T)= \mathcal P_{k}(S)$), then $S$ and  $T$. This can be proved in a manner similar to  Proposition~\ref{important1}.
%
\end{proof}

\begin{lemma}
\label{lemma_1time}
Let $k \geq 0$ and let $T$ be a 1-labeled rooted tree with $n$ nodes. All subsets $D_i(w) = \{ w \} \cup \{ x \in D_T(w) \mid d(w, x) \leq i \}$ and $L(D_i(w))$ for all nodes $w$ and $i\leq k$ can be computed in at most $2(k+1)n$ set operations.
\end{lemma}

\begin{proof}
Since $T$ is 1-labeled, we can identify a node of $T$ with its label. In this way, $D_i(w) =L(D_i(w))$ for all nodes $w$ and $i \leq k$. By ordering the $n$ labels, we represent each subset of labels (and each subset of nodes) as a $n$-bit 0-1 string, where the $i$-th bit is 1 if and only if the $i$-th label (node) is in the subset.

The statement is obvious in the case $k=0$, since $D_0(w) = \{ w \}$ and, clearly, all the $D_0(w)$ for $w \in V(T)$ can be computed in at most $2n$ set operations. We assume $k>0$ and prove the statement by induction as follows.

Assume that all the $D_{k-1}(w)$ for $w \in V(T)$ have been computed in at most $2kn$ set operations. Assume $w$ has $d_w$ children $u_1, u_2, \ldots, u_{d(w)}$. Then,
$$
D_k(w) = \{ w \} \cup \left( \cup_{i=1}^{d_w} D_{k-1}(u_i) \right)
$$
This implies that $D_k(w)$ for all $w$ can be computed from all $D_{k-1}(w)$
%
using  $\sum_{v \in V(T)} (1 + d_w) = 2n - 1$ set operations. In total, we can compute all subsets $D_i(w)$ $(i\leq k$ and $w\in V(T)$) in at most $2n - 1 + 2kn \leq 2(k + 1)n$ set operations.
\end{proof}

\begin{lemma}
\label{lemma_2time}
Let $k \geq 0$ and $T$ be a 1-labeled rooted tree with $n$ nodes. Using $L(D_i(w))$ for $w \in V(T), 0 \leq i \leq k$, we can compute $L(B_k(w))$ for all $w$ in $O(kn)$ set operations, where $B_k(w)$ is defined in Eqn.~(\ref{def_Bk}).
\end{lemma}

\begin{proof}
Since $T$ is a 1-labeled rooted tree, we identify a node with its label. In this way, we just need to show that $B_k(w)$ for all nodes $w$ can be computed in $O(kn)$ set operations. 

Let $r$ be the root of $T$. For any node $w \in V(T)$, let the unique path from $r$ to $w$ be
$$
w_0 = r, w_1, \ldots, w_t = w.
$$
Then, we have that
$$
B_k(w_{t}) = \cup^{\min(k, t)}_{i=0} D_{k-i}(w_{t-i}).
$$
Given the subsets $D_i(u)$ for all $i\leq k$ and $u\in V(T)$, the above formula implies that $B_k(w_{t})$ can be computed in at most $k$ set operations.
In total, we can compute all $B_k(w_{t})$ for all $w\in V(T)$
in $O(kn)$ set operations.
\end{proof}

\begin{proposition}
\label{time-1}
Let $S$ and $T$ be two 1-labeled trees with $n$ nodes and $k \geq 0$. Then, $d_{\kRF}(S, T)$ can be computed in $O(kn^2)$ time.
\end{proposition}

\begin{proof}
We first consider the rooted tree case. Let $S$ and $T$ be two 1-labeled rooted trees with $n$ nodes. Without loss of generality, we may assume that $S$ and $T$ are labeled on the same set $L$, with $|L| = n$. (Otherwise, we can consider them labeled on $L=L(S) \cup L(T)$, with $n \leq |L| \leq 2n$.) By Lemma~\ref{lemma_1time} and Lemma~\ref{lemma_2time}, we can compute $P_X(e, k)$ for all $e\in E(X)$ in $O(kn)$ set operations for $X = S, T$. Since each edge induces an ordered pair of label subsets and we represent each label subset using a $n$-bit string, we consider $P_X(e, k)$ as a $2n$-bit string. In this way, we sort all the edge-induced pairs of label subsets for each tree in $O(n^2)$ time by radix sort (that is, indexing) and then compute the symmetric difference of the two set of edge-induced pairs in $O(n^2)$ time. This concludes the proof.

In the unrooted case, we first root the trees at a leaf. In this way, we can compute all the edge-induced pairs of label subsets in the derived rooted trees in $O(kn^2)$ time. Since the edges induce unordered pairs of label subsets in the original trees, we rearrange the two label subsets obtained for an edge in such a way that the smallest label in the first subset is smaller than every label in the second one. After the rearrangement, we can radix-sort the edge-induced pairs and compute the $k$-RF score in $O(n^2)$ time.
\end{proof}


\begin{figure}
\centering
\includegraphics[scale=0.5]{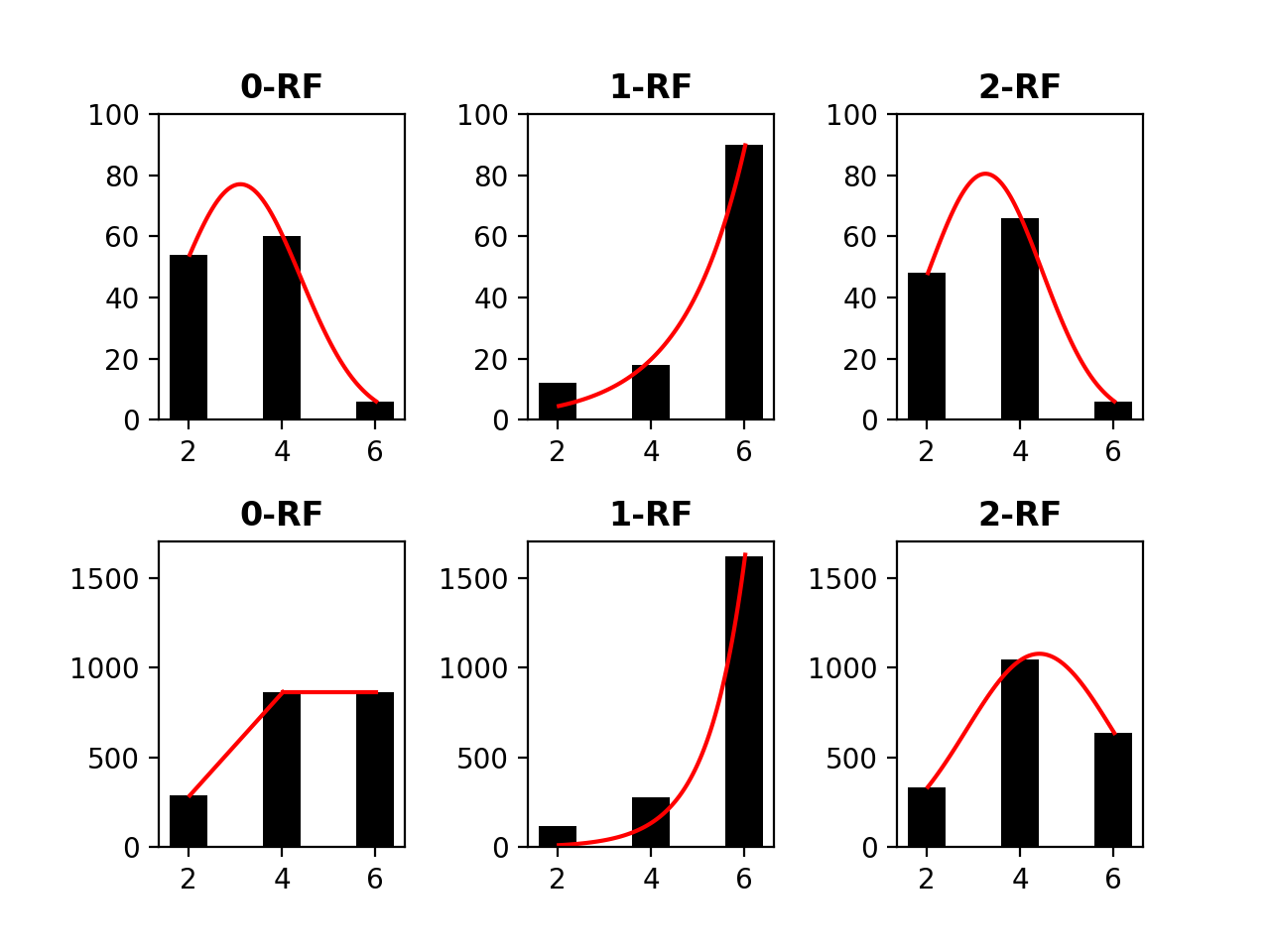}
\caption{\label{Frequency-2}
The frequency distributions of all pairwise $k$-Robinson-Foulds (RF)  scores in the space of 1-labeled unrooted (top row) and rooted (bottom row) 4-node trees for $k=0, 1, 2$.
In the bar-charts, the $x$-axis represents $k$-RF scores and the $y$-axis represents the number of tree pairs with a specific $k$-RF score.
}
\end{figure}

\begin{figure}[!t]
\centering
\includegraphics[width=\linewidth]{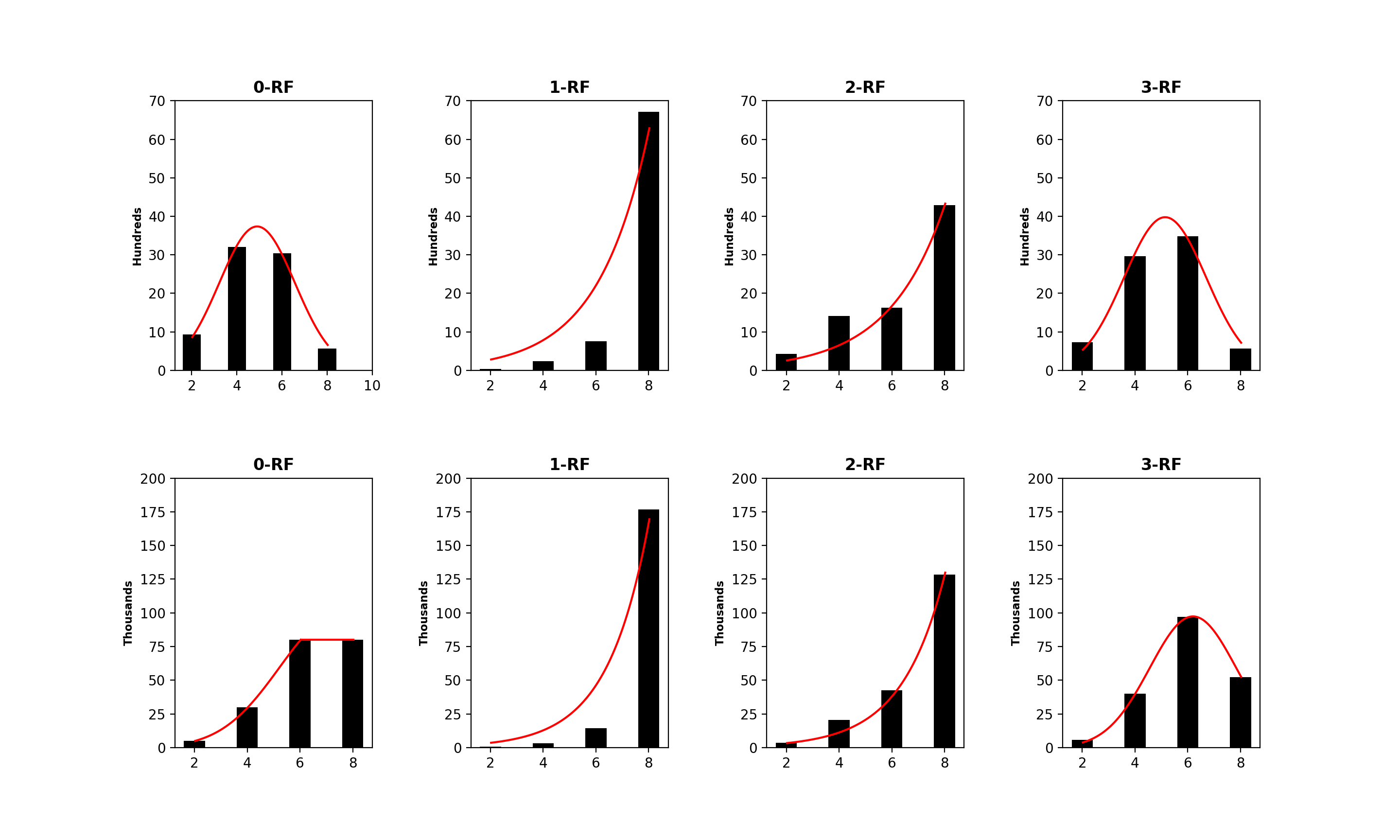}
\caption{\label{Frequency-3}
The frequency distributions of all pairwise $k$-Robinson-Foulds (RF) scores in the space of 1-labeled unrooted (top row) and rooted (bottom row) 5-node trees, where $k\leq 3$.
In the bar-charts, the $x$-axis represents $k$-RF scores and the $y$-axis represents the number of tree pairs with a specific $k$-RF score.
}
\end{figure}

\begin{figure}[!t]
\centering
\includegraphics[scale=0.2]{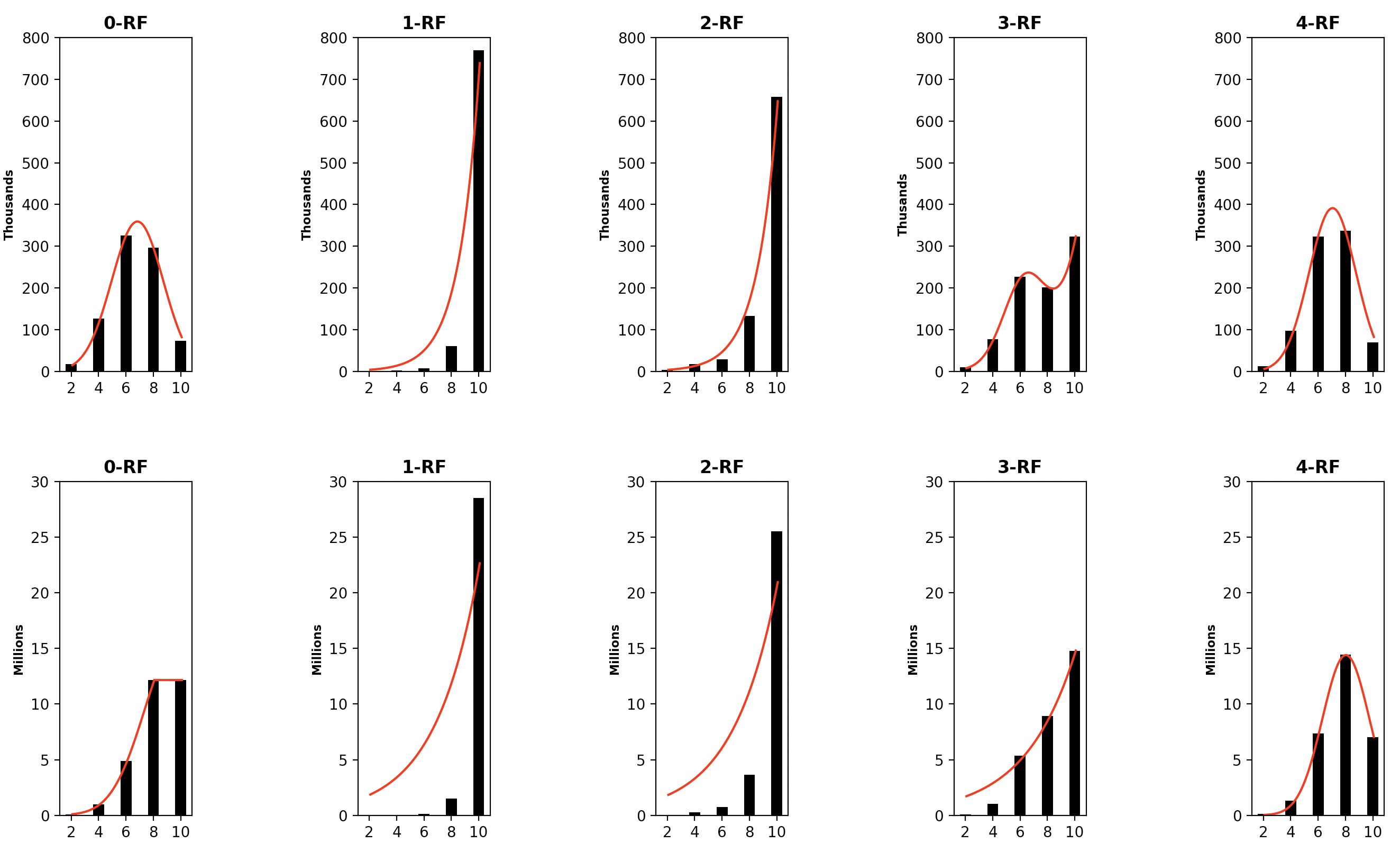}
\caption{\label{Frequency-4}
The frequency distributions of all pairwise $k$-Robinson-Foulds (RF) scores in the space of 1-labeled unrooted (top row) and rooted (bottom row) 6-node trees for $k\leq 4$.
In each bar-chart, the $x$-axis represents $k$-RF scores and the $y$-axis represents the number of tree pairs whose $k$-RF equals a given score.}
\end{figure}

\begin{figure}[!t]
\centering
\includegraphics[width=\linewidth]{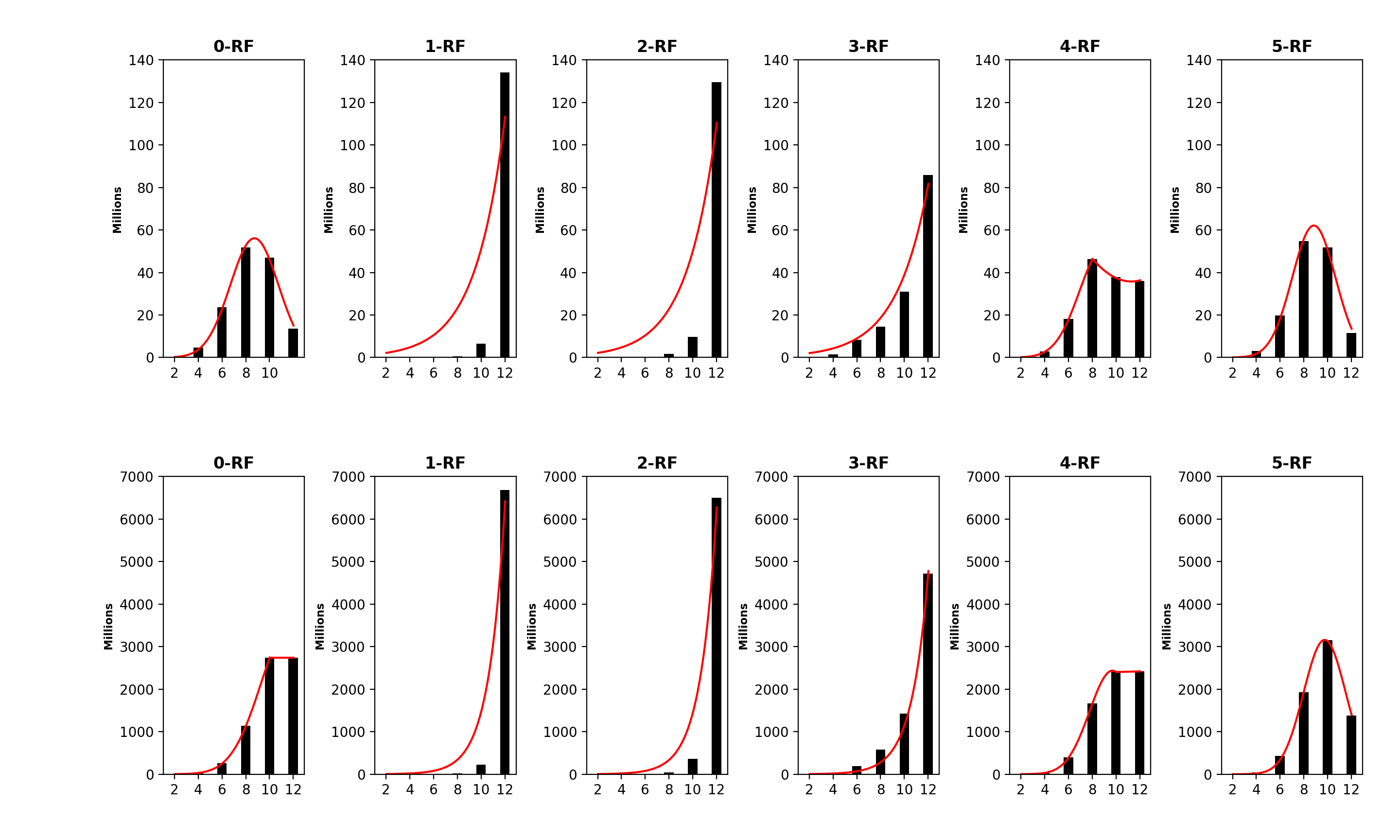}
\caption{\label{Frequency-5}
The frequency distributions of all pairwise $k$-Robinson-Foulds (RF) scores in the space of 1-labeled unrooted (top row) and rooted (bottom row) 7-node trees for $k\leq 5$.
In each bar-chart, the $x$-axis represents $k$-RF scores and the $y$-axis represents the number of tree pairs whose $k$-RF equals a given score.
}
\end{figure}

\subsection{The Distribution of \texorpdfstring{$k$-RF}{k-RF} Scores}

We examined the distribution of the $k$-RF dissimilarity scores for 1-labeled unrooted and rooted trees with the same label set and with different label sets.

The distribution of the frequency of the pairwise $k$-RF scores in the space of $n$-node 1-labeled unrooted and rooted trees  for $n$ from 4 to 7 are presented in Figures~\ref{Frequency-2} to \ref{Frequency-5}, respectively. For each $n$, it suffices to consider $k=0,...,n-2$. Recall that $(n-2)$-RF is actually the RF distance. The frequency distribution for the RF distance in the space of phylogenetic trees is known to be Poisson~\citep{3Steel}. It seems also true that the pairwise 0-RF and $(n-2)$-RF scores have a Poisson distribution in the space of $n$-node 1-labeled unrooted and rooted trees. However, the distribution of the pairwise $k$-RF scores is unlikely
Poisson when $k = 1, 2, 3$ and $k\neq n-2$.

We examined 1,679,616 (respectively, 60,466,176) pairs of 6-node 1-labeled unrooted (respectively, rooted) trees such that the trees in each pair have $c$ common labels, with $c = 3, 4, 5$. Table~\ref{table0} shows that the majority of pairs have a largest dissimilarity score of 10.

\begin{table}[b!]
\caption{\label{table0} The number of pairs of 1-labeled 6-node unrooted (top) and rooted (bottom) trees that have $c$ labels in common and have 1-Robinson-Foulds (RF) score $d$ for $c = 3, 4, 5$ and $d = 2, 4, 6, 8, 10$.}
\begin{flushleft}
\begin{tabular}{c|rrrrr} \hline
1-RF & ~~2 & 4 & 6 & 8  & 10 \\ \hline
3&0&0&0&3,072&1,676,544\\ \hline
4&0&0&432&16,800&1,662,384\\ \hline
5&0& 340& 3,720&53,100&1,622,456\\ \hline
\end{tabular}
\end{flushleft}
\begin{flushright}
\begin{tabular}{c|rrrrr} \hline
1-RF & ~~2 & 4 & 6 & 8 & 10 \\ \hline
3&0 &0&0&79,872&60,386,304\\ \hline
4& 0&0&7,776& 419,136&60,039,264\\ \hline
5&0 &4,080& 65,760& 1,310,880&59,085,456\\ \hline
\end{tabular}
\end{flushright}
\end{table}

\section{A Generalization to Multiset-Labeled Trees}
\label{sec 6}

In this section, we extend the measures introduced in Section~\ref{sec 3} to multiset-labeled unrooted and rooted trees.

\subsection{Multisets and Their Operations}

A multiset is a collection of elements in which an element $x$ can occur one or more times~\citep{7Jurgensen}. The set of all distinct elements appearing in a multiset $A$ is denoted by $\supp(A)$. 
In this paper, we simply represent $A$ by the monomial $x_1^{m_A(x_1)} \ldots x_n^{m_A(x_n)}$ if $\supp(A)=\{x_1, x_2, \cdots, x_n\}$, where $x_i^{1}$ is simplified to $x_i$ for each $i$.

Let $A$ and $B$ be two multisets. We say $A$ is a sub-multiset of $B$, denoted by $A \subseteq_m B$, if for every $x \in \supp(A)$, $m_A(x) \leq m_B(x)$. In addition, we say that $A = B$ if $A \subseteq_m B$ and $B \subseteq_m A$. Furthermore, the union, sum, intersection, difference, and symmetric difference of $A$ and $B$ are respectively defined as follows:
\begin{itemize}
\item $A \cup_m B = \left\{ x^{\max\{m_A(x), m_B(x)\}} \mid x \in \supp(A) \cup \supp(B) \right\}$;
\item $A \uplus_m B = \left\{ x^{m_A(x)+ m_B(x)} \mid x \in \supp(A) \cup \supp(B) \right\}$; 
\item $A \cap_m B = \left\{ x^{\min\{m_A(x), m_B(x)\}} \mid x \in \supp(A) \cap \supp(B) \right\}$;
\item $A \setminus_m B = \left\{ x^{ m_A(x) - m_B(x)} \mid x \in \supp(A) : m_A(x) > m_B(x) \right\}$; 
\item $A \triangle_m B = (A \cup_m B) \setminus_m (A \cap_m B)$; 
\end{itemize}
where $m_X(x)$ is defined as 0 if $x \notin \supp(X)$ for $X = A, B$.

Let $L$ be a set and $\mathbb{P}_{m}(L)$ be the set of all sub-multisets on $L$. A tree $T$ is labeled with the sub-multisets of $L$ if $T$ is equipped with a function $\ell: V(T) \to \mathbb{P}_{m}(L)$ such that $\cup_{v \in V(T)} \supp(\ell(v)) = L$ and $\ell(v) \neq \emptyset$, for every $v \in V(T)$. We call such a tree as a multiset-labeled tree. For $C \subseteq V(T)$ and $x \in L$, we define $L_m(C)$ and $m_{T}(x)$ as follows:
\begin{eqnarray}
L_m(C) &=& \uplus_{v \in C} \ell(v); \label{set_label}\\
m_{T}(x) &=& \sum_{v \in V(T)} m_{\ell(v)}(x). 
\end{eqnarray}

\subsection{The \texorpdfstring{$k$-RF}{k-RF} for Multiset-Labeled Trees}
\label{subsec 5.1}

Let $T$ be a multiset-labeled tree. Then, each edge $e=(u, v)$ of $T$ induces a pair of multisets
\begin{equation}
P_{T}(e) = \left\{ L_m(B_e(u)), L_m(B_e(v)) \right\},
\label{multi_partition}
\end{equation}
where $L_m( )$ is defined in Eqn.~(\ref{set_label}), and $B_e(u)$  in Eqn.~(\ref{eqn2_sec31}). Note that Eqn.~(\ref{multi_partition}) is obtained from Eqn.~(\ref{eqn1_sec31}) by replacing $L()$ with $L_m()$.

\begin{remark}
In a multiset-labeled tree $T$, two edges may induce the same multi-set pair as shown in Figure~\ref{fig:aa}. Hence, $\calP(T)$ in Eqn.~(\ref{eqn3_sec31}) is a multiset in general.
\end{remark}

\begin{figure} 
\centering
\includegraphics[scale=0.6]{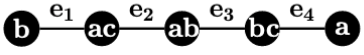}
\caption{\label{fig:aa} 
Two  multiset-labeled trees used to show that different edges can give the same label multi-subset pair. Here, $P_{T}(e_2) = P_{T}(e_3) = \{ abc, a^2b^2c \}$.}
\end{figure}

We use Eqn.~(\ref{multi_partition}), Eqn.~(\ref{eqn3_sec31}) and Eqn.~(\ref{eqn4_sec31}) to define the RF-distance for multiset-labeled trees by replacing $\triangle$ with $\triangle_m$ in Eqn.~(\ref{eqn4_sec31}).

Let $k \geq 0$. We use Eqn.~(\ref{eqn5_sec32}), Eqn.~(\ref{eqn6_sec32}), and Eqn.~(\ref{eqn7_sec32}) to define the $k$-RF  for multiset-labeled trees by replacing $L()$ with $L_m()$ in Eqn.~(\ref{eqn5_sec32}) and replacing $\triangle$ with $\triangle_m$ in Eqn.~(\ref{eqn7_sec32}).

\begin{example}
\label{example_4}
Consider the multiset-labeled trees $S$, $\acute{S}$, and $T$ in Figure~\ref{fig:1}. ${\mathcal P}_k(T), {\mathcal P}_k(S)$ and ${\mathcal P}_k(\acute{S})$ for $k = 0, 1, \infty$ are summarized in Table~\ref{table1}. We obtain:
$$
\begin{array}{lll}
  d_{0\mbox{\tiny -RF}}(T, \acute{S}) = 2;
& d_{\OneRF}(T, \acute{S}) = 6;
& d_{\mbox{\tiny RF}}(T, \acute{S}) = 12;\\ 
  d_{0\mbox{\tiny -RF}}(S, \acute{S}) = 10;
& d_{\OneRF}(S, \acute{S}) = 12;
& d_{\mbox{\tiny RF}}(S, \acute{S}) = 12.\\
\end{array}
$$
It is not hard to see that both  $d_{0\mbox{\tiny -RF}}(T, \acute{S})$ and $d_{\OneRF}(T, \acute{S})$ reflect the local similarity of the two multiset-labeled trees better than $d_{\mbox{\tiny RF}}(T, \acute{S})$.
\end{example}

\begin{figure}[t!]
\centering
\includegraphics[scale=0.5]{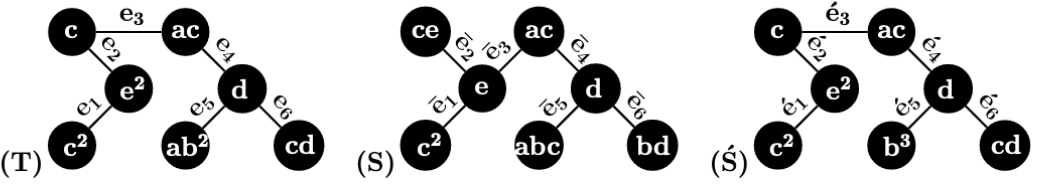}
\caption{\label{fig:1}Three multiset-labeled trees in Example~\ref{example_4}.}
\end{figure}

\begin{table}[b!]
\caption{\label{table1} Edge-induced unordered pairs of multisets in the three trees in Fig.~\ref{fig:1} for $k = 0, 1, \infty$.}
$$
\begin{array}{clll} \hline
\textrm{Tree} & \mathcal{P}_0(\;) & \mathcal{P}_1(\;) & \mathcal{P}_\infty(\;) \\ \hline
T & \{c^2, e^2\} & \{c^2, ce^2\} & \{a^2b^2c^3d^2e^2, c^2\} \\
& \{c, e^2\} & \{ab^2cd^2, ac^2\} & \{a^2b^2c^3d^2, c^2e^2\} \\
& \{ac, c\} & \{ac^2, c^2e^2\} & \{a^2b^2c^2d^2, c^3e^2\} \\
& \{ac, d\} & \{ab^2, ac^2d^2\} & \{ab^2cd^2, ac^4e^2\} \\
& \{ab^2, d\} & \{acd, ce^2\} & \{ab^2, ac^5d^2e^2\} \\
& \{cd, d\} & \{a^2b^2cd, cd\} & \{a^2b^2c^4de^2, cd\} \\ \hline
S & \{c^2, e\} & \{ac^2e^2, c^2\} & \{a^2b^2c^3d^2e^2, c^2\} \\
& \{ce, e\} & \{a^2bc^2d, bd\} & \{a^2b^2c^2d^2, c^3e^2\} \\
& \{ac, e\} & \{ab^2cd^2, ace\} & \{ab^2cd^2, ac^4e^2\} \\
& \{ac, d\} & \{ac^3e, ce\} & \{a^2b^2c^4d^2e, ce\} \\
& \{abc, d\} & \{acd, c^3e^2\} & \{abc, abc^4d^2e^2\} \\
& \{bd, d\} & \{abc, abcd^2\} & \{a^2bc^5de^2, bd\} \\ \hline
\acute{S} & \{c^2, e^2\} & \{c^2, ce^2 \}, & \{ab^3c^3d^2e^2, c^2\} \\
& \{c, e^2\} & \{ac^2, c^2e^2\} & \{ab^3c^3d^2, c^2e^2\} \\
& \{ac, c\} & \{acd, ce^2\} & \{ab^3c^2d^2, c^3e^2 \} \\
& \{ac, d\} & \{ac^2d^2, b^3\} & \{ac^4e^2, b^3cd^2 \} \\
& \{b^3, d\} & \{ac^2, b^3cd^2\} & \{ac^5e^2d^2, b^3\} \\
& \{cd, d\} & \{ab^3cd, cd\} & \{ab^3c^4e^2d, cd\} \\ \hline
\end{array}
$$
\end{table}

\subsection{The \texorpdfstring{$k$-RF}{k-RF} for Multiset-Labeled Rooted Trees}
\label{subsec 5.4}

Let $k \geq 0$ be an integer. We use Eqn.~(\ref{eqn11_sec34}), Eqn.~(\ref{eqn12_sec34}), and Eqn.~(\ref{eqn13_sec34}) to define $k$-RF for multiset-labeled rooted trees by replacing $L()$ with $L_m()$ in Eqn.~(\ref{eqn11_sec34}) and replacing $\triangle$ with $\triangle_m$ in Eqn.~(\ref{eqn13_sec34}).


\begin{proposition}
\label{metric}
Let $k \geq 0$ be an integer. The $k$-RF  satisfies the non-negativity, symmetry, and triangle inequality conditions. Hence, $k$-RF is a pseudometric for each $k$ in the space of multiset-labeled (rooted) trees.
\end{proposition}
\begin{proof} The non-negativity and symmetry conditions follow from the definition of the 
$k$-RF. The triangle inequality condition is proved as follows.

Let $T_{1}$, $T_{2}$, and $T_{3}$ be three multiset-labeled trees. We need to show: 
\begin{eqnarray*}
&&d_{\kRF}(T_{1},T_{2}) \leq d_{\kRF}(T_{1},T_{3}) + d_{\kRF}(T_{3},T_{2}).
\end{eqnarray*}
Note that  $\mathcal P_{k}(X)$ 
denotes the multiset of edge-induced order pairs of sub-multisets in $X$ for $X=T_1, T_2, T_3$. 

If $x^{m(x)} \in \mathcal P_{k}(T_1) \triangle_m \mathcal P_{k}(T_2)$, we have either $x^{m(x)} \in \mathcal P_{k}(T_1) \setminus_m \mathcal P_{k}(T_2)$ or $x^{m(x)} \in \mathcal P_{k}(T_2) \setminus_m \mathcal P_{k}(T_1)$. Assume $x^{m(x)} \in \mathcal P_{k}(T_1) \setminus_m \mathcal P_{k}(T_2)$. Then,  $m_{\mathcal P_{k}(T_1)}(x) > m_{\mathcal P_{k}(T_2)}(x)$. If $x \notin \textrm{Supp}(\mathcal P_{k}(T_3) \setminus_m \mathcal P_{k}(T_2))$, we have $m_{\mathcal P_{k}(T_1)}(x) > m_{\mathcal P_{k}(T_2)}(x) \geqslant m_{\mathcal P_{k}(T_3)}(x)$.  This implies that $x \in \mbox{Supp}(\mathcal P_{k}(T_1)\setminus_m \mathcal P_{k}(T_3))$ and
$
m_{\mathcal P_{k}(T_1) \setminus_m \mathcal P_{k}(T_3)}(x) = m_{\mathcal P_{k}(T_1)}(x) - m_{\mathcal P_{k}(T_3)}(x) \geqslant m_{\mathcal P_{k}(T_1)}(x)-m_{\mathcal P_{k}(T_2)}(x) = m(x).
$
Thus, 
$
m(x)\leqslant m_{\mathcal P_{k}(T_1) \triangle_m \mathcal P_{k}(T_3)}(x) + m_{\mathcal P_{k}(T_3) \triangle_m \mathcal P_{k}(T_2)}(x).
$

On the other hand, if $x \in \textrm{Supp}(\mathcal P_{k}(T_3) \setminus_m \mathcal P_{k}(T_2))$ and $m_{\mathcal P_{k}(T_3)}(x) \geq m_{\mathcal P_{k}(T_1)}(x)$, we have:
\begin{eqnarray*}
m_{\mathcal P_{k}(T_3) \setminus_m \mathcal P_{k}(T_2)}(x) &=& m_{\mathcal P_{k}(T_3)}(x) - m_{\mathcal P_{k}(T_2)}(x) \\
&\geq&  m_{\mathcal P_{k}(T_1)}(x) - m_{\mathcal P_{k}(T_2)}(x) = m(x).
\end{eqnarray*}
If $x \in \textrm{Supp}(\mathcal P_{k}(T_3) \setminus_m \mathcal P_{k}(T_2))$ and $m_{\mathcal P_{k}(T_3)}(x) < m_{\mathcal P_{k}(T_1)}(x)$, we have $m_{\mathcal P_{k}(T_1)}(x) > m_{\mathcal P_{k}(T_3)}(x) > m_{\mathcal P_{k}(T_2)}(x)$, implying $x \in \textrm{Supp}(\mathcal P_{k}(T_1) \setminus_m \mathcal P_{k}(T_3))$. Thus, we have: 
\begin{eqnarray*}
m(x) &=& m_{\mathcal P_{k}(T_1) \setminus_m \mathcal P_{k}(T_3)}(x) + m_{\mathcal P_{k}(T_3) \setminus_m \mathcal P_{k}(T_2)}(x) \\
&\leq& m_{\mathcal P_{k}(T_1)\triangle_m \mathcal P_{k}(T_3)}(x) + m_{\mathcal P_{k}(T_3) \triangle_m \mathcal P_{k}(T_2)}(x).
\end{eqnarray*}

Lastly,  if $x^{m(x)} \in \mathcal P_{k}(T_2) \setminus_m \mathcal P_{k}(T_1)$,  we can obtain the same result.

In summary, we have:
\[
\textrm{Supp}(\mathcal P_{k}(T_1) \triangle_m \mathcal P_{k}(T_2)) \subseteq  \textrm{Supp}(\mathcal P_{k}(T_1) \triangle_m \mathcal P_{k}(T_3)) \cup \textrm{Supp}(\mathcal P_{k}(T_3) \triangle_m \mathcal P_{k}(T_2)).
\]
In addition, for each $x \in \textrm{Supp}(\mathcal P_{k}(T_1) \triangle_m \mathcal P_{k}(T_2))$, we have:
\[
m_{\mathcal P_{k}(T_1) \triangle_m \mathcal P_{k}(T_2)}(x) \leq m_{\mathcal P_{k}(T_1) \triangle_m \mathcal P_{k}(T_3)}(x) + m_{\mathcal P_{k}(T_3)\triangle_m \mathcal P_{k}(T_2)}(x).
\]
Therefore, we have:
\[
|\mathcal P_{k}(T_1) \triangle_m \mathcal P_{k}(T_2)| \leq |\mathcal P_{k}(T_1) \triangle_m \mathcal P_{k}(T_3)| + |\mathcal P_{k}(T_3) \triangle_m \mathcal P_{k}(T_2)|,
\]
that is,  the triangle  inequality holds.

For multiset-labeled rooted trees, the proof is similar and hence omitted.
\end{proof} 

\begin{figure} [t!]
\centering
\includegraphics[scale=0.7]{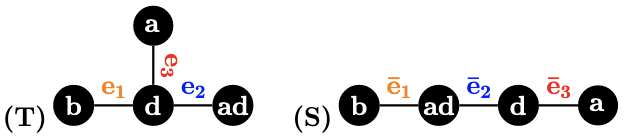}
\caption{\label{fig:similartrees}
Two distinct  multiset-labeled trees $S$ and $T$ satisfy that  ${\cal P}_{2}(S) = {\cal P}_{2}(T) =\{\{a^2d^2, b\}, \{abd,ad\}, \{a,abd^2\} \}$, showing that
2-RF score can be 0 even for distinct trees.
}
\end{figure}

\begin{remark}
For multiset-labeled trees,  $d_{\kRF}(S,T) = 0$ does not imply $S$ and $T$ are identical, as shown in Fig.~\ref{fig:similartrees}. 
\end{remark}

\begin{proposition}
\label{prop5}
Let $k \geq 0$ and $S$ and $T$ be two (rooted) trees whose nodes are labeled by $L(S)$ and $L(T)$, respectively. Then, $d_{\kRF}(S, T)$ can be computed in $O((k+B)D(\vert V(S)\vert +\vert V(T)\vert)$ time, where
$B$ is the maximum multiplicity of a label appearing in $\{P_T(e, k) \mid e\in V(T)\}\cup \{ P_S(e, k) \mid e\in V(S)\}$ and 
$D=\vert \mathit{Supp}(L(S))\cup \mathit{Supp}(L(T))\vert$.
\end{proposition}

\begin{proof}
An algorithm for the 1-labeled case can be modified as follows for computing $k$-RF on multiset-labeled rooted and unrooted trees:
\begin{itemize}
\item Represent each label multiset as a $D$-dimensional vector, in which the integer at position $j$ is the multiplicity of the $j$-th label.
Computing all edge-induced pairs in both trees takes  $O(k (\vert E(S)\vert +\vert E(T)\vert))$ set operations.   Each set operation takes $D$ integer operations.
\item Radix-sort all the edge-induced pairs for $S$ and $T$ in $O(D(\vert E(S)\vert + B))$ and $O(D(\vert E(T)\vert + B))$ integer operations, respectively.
\item Compute the symmetric difference of the set of the edge-induced pairs in the two input trees in $\vert E(S)\vert +\vert E(T)\vert$ set operation.  Each set operation takes $D$ integer operations.
\end{itemize}
In summary, one can compute 
$d_{\kRF}(S, T)$ using $O((k+B)D(\vert V(S)\vert +\vert V(T)\vert)$
integer operations, as $\vert E(S)\vert =\vert V(S)\vert -1$.
\end{proof}


\begin{figure}[t!]
\centering
\includegraphics[scale=0.30]{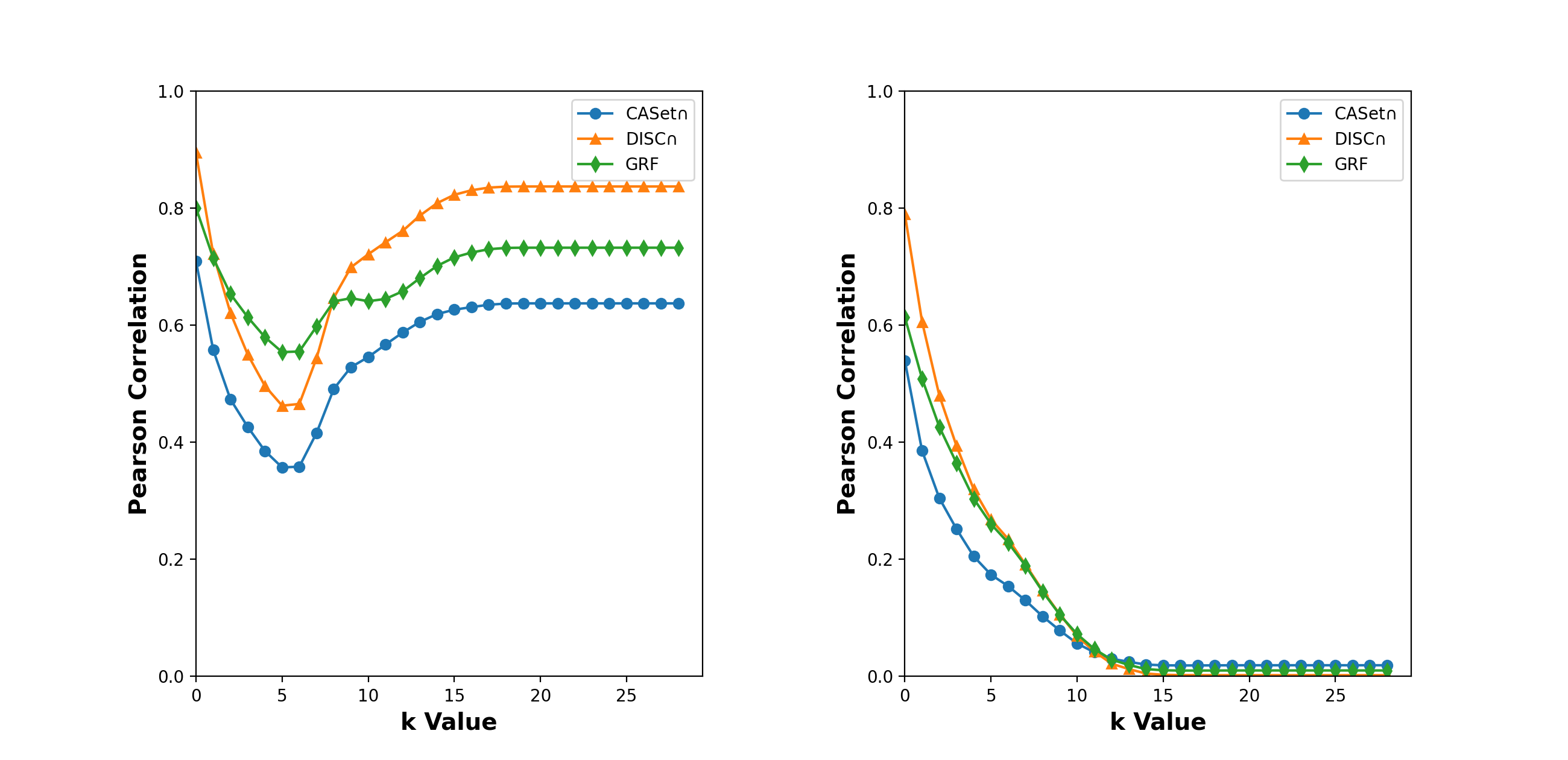}
\caption{\label{fig7:correlation}
Pearson correlation of the $k$-Robinson-Foulds (RF) with CASet$\cap$, DISC$\cap$, and GRF. The analyses were conducted on rand rooted trees with the same label set (left) and with different but overlapping label sets (right) that were reported in \citet{1John_Zhang2019}. 
The Pearson correlation became constant for $k \geq 19$ in the range $k$-RF becomes RF.
CASet$\cap$: Common Ancestor Set distance;
DISC$\cap$: Distinctly Inherited Set Comparison distance~\cite{4Dinardo}. 
GRF:  Generalized RF distance~\cite{llabres.ea:2021}.
}
\end{figure}

\subsection{Correlation of the \texorpdfstring{$k$-RF}{k-RF} and the Other Measures}
\label{subsec 5.6}



Let $T$ and $S$ be two 1-labeled rooted trees with the same label set $X$.  Again, we identify the nodes with their labels in the two trees. For any two subset $X'$ and $X''$ of $X$, we use $d_{J}(X', X'')$ to denote their Jaccard distance. 
The CASet$\cap$ distance between  $T$ and $S$ is defined to be
the average 
$d_{J}(A_T(i)\cap A_T(j), A_S(i)\cap A_S(j))$
of a pair of nodes $i$ and $j$, whereas the DISC$\cap$ distance between $T$ and $S$ is the average  
$d_J(A_T(i)\setminus A_T(j), A_S(i)\setminus A_S(j))$ of an order pair $(i, j)$ of nodes \cite{4Dinardo}.

Using the Pearson correlation, we compared the $k$-RF with CASet$\cap$, DISC$\cap$, and GRF~\citep{llabres.ea:2020} in the space of set-labeled trees for different $k$ from 0 to 28.

Firstly,  we conducted the correlation analysis in the space of mutation trees with the same label set. Using a method reported by~\citet{1John_Zhang2019}, we generated a simulated dataset containing 5,000 rooted trees in which the root was labeled with 0 and the other nodes were labeled by the disjoint subsets of $\{1,2,\ldots,29\}$, where the trees might have different number of nodes.  Using all $\binom{5,000}{2}$ pairwise scores for CASet$\cap$, DISC$\cap$, GRF and $k$-RF, we conducted the Pearson correlation analysis of $k$-RF with the other three (left panel,  Fig.~\ref{fig7:correlation}).

Our results show that CASet$\cap$, DISC$\cap$ and GRF were all positively correlated with $k$-RF. We observed the following facts:
\begin{itemize} 
\item The GRF and $k$-RF had the largest Pearson correlation for each $k < 8$, whereas the DISC$\cap$ and $k$-RF  had the largest Pearson correlation for each $k \geq 8$. 
\item The 5-RF and 6-RF were less correlated to CASet$\cap$, DISC$\cap$ and GRF than other $k$-RF. 
\item The Pearson correlation between $k$-RF and  CASet$\cap$ (respectively, DISC$\cap$) increased when $k$ went from 6 to 15.
\end{itemize}
Secondly,  we  conducted the Pearson correlation analysis on the trees with different but overlapping label sets. The dataset was generated by the same method and was a union of 5 groups of rooted trees, each of which contained 200 trees over the same label set. We computed the dissimilarity scores for each tree in the first family and each tree in other groups and then computed the Pearson correlation between different measures. Again, all the dissimilarity measures were positively correlated, but less correlated than in the first case; see Fig~\ref{fig7:correlation} (right). 
This observation could be the result of the fact that difference in label sets of two trees makes  their $k$-RF score at least $k+1$. However, the difference does not strongly contribute to the other distances because DISC$\cap$ and CASet$\cap$ consider the intersection of label sets (see ~\citep{4Dinardo}), and GRF considers the intersection of clusters.

The right dotplot of Fig.~\ref{fig7:correlation} shows that the $k$-RF and DISC$\cap$ had the largest Pearson correlation for $k$ from 1 to 9, and the $k$-RF and the CASet$\cap$ had the largest Pearson correlation for $k \geq 10$. Moreover, all the Pearson correlations decreased when $k$ changed from 1 to 15. This trend was not observed in the first case.
This decreasing trend could be the result of the fact that difference in label sets contributes to $k$-RF more as $k$ increases.

\section{Clustering Trees with the 
\texorpdfstring{$k$-RF}{k-RF}}
\label{sec_5}

A test was designed to demonstrate which of the $k$-RF, CASet$\cap$, DISC$\cap$, and GRF is good at clustering labeled trees.

We generated randomly 5 tree families each containing 50 trees using the program reported by~\citet{1John_Zhang2019}. The nodes were labeled by the subsets of a 30-label set in the trees of each family. The label sets used in different tree families were different, but overlapping. As the nodes were labeled by disjoint subsets, each different label between the label sets of two trees induces at least $d$ different pairs, where $d$ is the degree of the node with the label. Thus, a large number of different elements between the label sets could make the trees more distinguishable by the $k$-RF. Therefore, the label sets used for the different tree families differed in only one label.

We computed the pairwise dissimilarity scores for all 250 trees in the five groups using each measure; we then clustered the 250 trees into $c$ clusters using the $K$-means algorithm, where $c$ ranges from $2$ to $57$. The clustering results were assessed  using the Silhouette score~\citep{kaufman2009finding}.

\begin{figure} 
\centering
\includegraphics[width=0.9\textwidth]{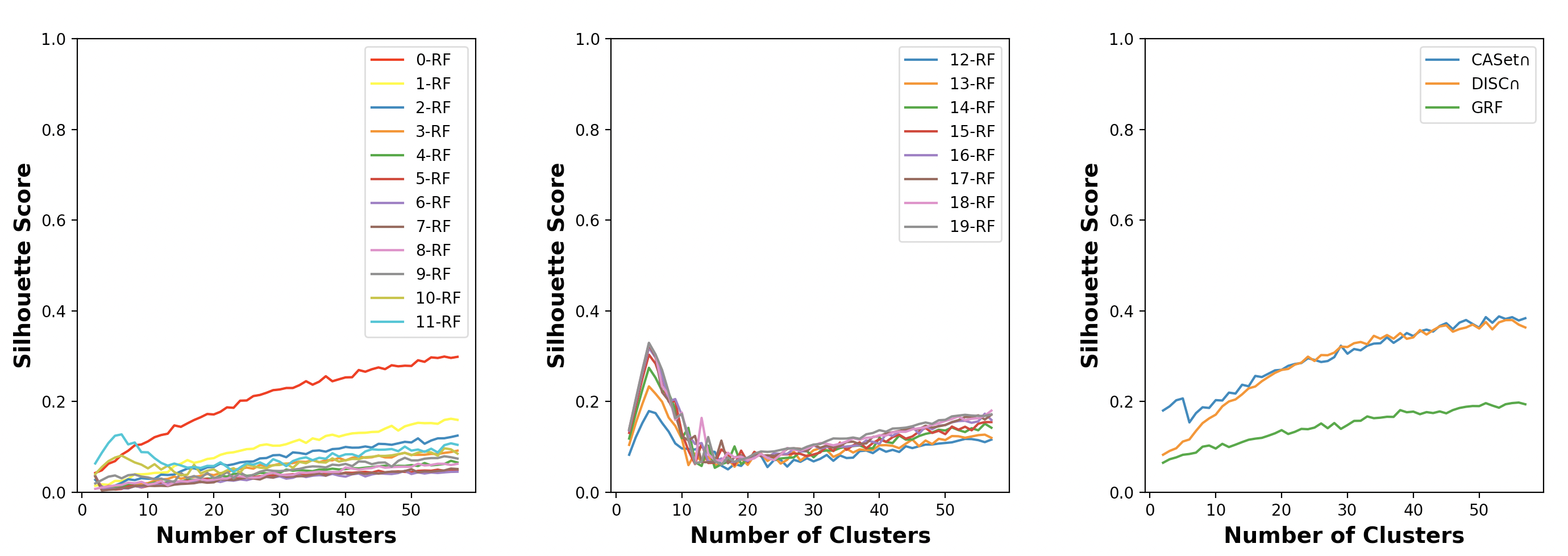}
\caption{\label{sil}Silhouette scores of clustering 250 rooted trees with $k$-RF for $0\leq k\leq 11$ (left) and  $12\leq k\leq 19$ (middle) and with CASet$\cap$, DISC$\cap$, and GRF (right).
RF: Robinson-Foulds.
CASet$\cap$: Common Ancestor Set distance;
DISC$\cap$: Distinctly Inherited Set Comparison distance~\cite{4Dinardo}. 
GRF:  Generalized RF distance~\cite{llabres.ea:2021}.}
\end{figure}

As Fig.~\ref{sil} illustrates, neither of the CASet$\cap$, DISC$\cap$, and GRF distances were able to recognize the exact number of families. However, CASet$\cap$ had the highest Silhouette score when the number of clusters was 5, compared to DISC$\cap$, GRF, and the $k$-RF for $k\leq 12$. In addition, the figure shows that the $k$-RF  could recognize the correct number of families when $k$ ranges from 12 to 19. Moreover, the Silhouette score of the $k$-RF  increased when $k$ increased from $8$ to $19$. 
This interesting observation 
may stem from the fact that as $k$ increases, the number of pairs of trees achieving the highest possible $k$-RF score also increases, thereby enhancing the recognizability of families. It's worth noting that such pairs are guaranteed to exist when $k$ is larger than the minimum diameter of the trees, which is 8 in our case.

\section{Conclusions}
\label{sec 7}

The development of an efficient and robust measure for the comparison of labeled trees is important. In this paper, we have proposed a novel variant of dissimilarity metrics, namely the $k$-RF, tailored for labeled trees. The $k$-RF facilitates the analysis of local structures in labeled trees, accommodating nodes labeled with (not necessarily the same) multisets. Significantly, these metrics find practical applicability in mutation trees used in cancer research.

The RF distance is succinctly expressed as 
$(n-1)$-RF within the space of labeled
trees with $n$ nodes.  By setting $k$ to a value smaller than $n-1$,  the $k$-RF metric can capture analogous local regions in two labeled trees. Notably, for every $k$,  the $k$-RF is a pseudometric for multiset-labeled trees and becomes a metric in the space of 1-labeled trees. However,  
the distribution of pairwise $k$-RF scores
in the space of 1-labeled unrooted (or rooted) trees conforms to a Poisson distribution
specifically for $k=n-2$, and unlikely have the same trend for other values of $k\geq 1$.



We verified  the $k$-RF measures through a comprehensive comparison with CASet, DISC 
(\cite{4Dinardo}) and GRF (\cite{llabres.ea:2021}) on randomly labeled trees generated by a house-made program (\cite{1John_Zhang2019}). Our findings revealed a consistent positive correlation between $k$-RF and each of the other three measures for every value of $k$. Notably, the correlation values exhibited a tendency to be higher when the measures were applied to assess mutation trees with identical label sets. Furthermore, our study underscored the superior clustering capabilities of $k$-RF compared to the three mentioned measures.

We would like to emphasize that selecting an appropriate $k$-RF in practical applications lacks a universal rule of thumb, primarily due to a shortage of experience in this domain. Perhaps a judicious approach involves choosing a suitable $k$-RF by carefully considering the topological similarity among the trees under consideration.

Future work includes how to apply the $k$-RF  to designing tree inference algorithms like GraPhyC~\citep{9Govek} and also how to infer the exact frequency distribution of the $k$-RF for each $k \geq 1$. It is also interesting to investigate the generalization of RF-distance for clonal trees~\citep{llabres.ea:2020}.

The computer program for the $k$-RF can be downloaded from \url{https://github.com/Elahe-khayatian/k-RF-measures.git}.

\section*{Acknowledgments}
\addcontentsline{toc}{section}{Acknowledgments}

The authors would like to thank the anonymous reviewer for providing helpful suggestions 
and comments to our first submission of the work.
This research was partially supported by the the Ministerio de Ciencia e Innovaci\'{o}n (MCI), 
the Agencia Estatal de Investigaci\'{o}n (AEI) and the European Regional Development Funds (ERDF) 
through project METACIRCLE PID2021-126114NB-C44, also supported by the European Regional Development Fund (FEDER), 
by the Agency for Management of University and Research Grants (AGAUR) through grant 2017-SGR-786 (ALBCOM), 
and by Singapore MOE Tier 1 grant R-146-000-318-114.


\end{document}